\documentclass{article}

\usepackage{arxiv}

\usepackage[utf8]{inputenc} % allow utf-8 input
\usepackage[T1]{fontenc}    % use 8-bit T1 fonts
\usepackage{hyperref}       % hyperlinks
\usepackage{url}            % simple URL typesetting
\usepackage{booktabs}       % professional-quality tables
\usepackage{amsfonts} 
\usepackage{amsmath,amsthm} % blackboard math symbols
\usepackage{nicefrac}       % compact symbols for 1/2, etc.
\usepackage{microtype}      % microtypography
\usepackage{lipsum}		% Can be removed after putting your text content
\usepackage{graphicx}
\usepackage{natbib}
\usepackage{doi}

\newcommand{\reais}{I\!\! R}

\usepackage{hyperref}
\usepackage{mathtools} 
\usepackage{rotating}
\usepackage{IEEEtrantools}
\usepackage{comment}
\usepackage{algpseudocode, algorithmicx,algorithm}
\usepackage{adjustbox}
\usepackage{subfig}
\usepackage{tabularx,booktabs}
\usepackage{rotating}

\newtheorem{theorem}{Theorem}[section]

\newtheorem{lemma}[theorem]{Lemma}
\newtheorem{corollary}[theorem]{Corollary}%

\title{A note on improving the search of optimal prices in envy-free perfect matchings}

\author{ Marcos Salvatierra\\ Normal Superior School \\ Amazonas State University \And Juan G. Colonna \\ Institute of Computing \\ Federal University of Amazonas \\ \And Mario Salvatierra Jr. \\ Institute of Computing \\ Federal University of Amazonas \\ \And Alcides de C. Amorim Neto \\ Normal Superior School \\ Amazonas State University
}

\date{}

\hypersetup{
pdftitle={A template for the arxiv style},
pdfsubject={q-bio.NC, q-bio.QM},
pdfauthor={David S.~Hippocampus, Elias D.~Striatum},
pdfkeywords={First keyword, Second keyword, More},
}

\begin{document}
\maketitle

\begin{abstract}
	We present a method for finding envy-free prices in a combinatorial auction where the consumers' number $n$ coincides with that of distinct items for sale, each consumer can buy one single item and each item has only one unit available. This is a particular case of the {\it unit-demand envy-free pricing problem}, and was recently revisited by Arbib et al. (2019). These authors proved that using a Fibonacci heap for solving the maximum weight perfect matching and the Bellman-Ford algorithm for getting the envy-free prices, the overall time complexity for solving the problem is $O(n^3)$. We propose a method based on dynamic programming design strategy that seeks the optimal envy-free prices by increasing the consumers' utilities, which has the same cubic complexity time as the aforementioned approach, but whose theoretical and empirical results indicate that our method performs faster than the shortest paths strategy, obtaining an average time reduction in determining optimal envy-free prices of approximately 48\%.
\end{abstract}

% keywords can be removed

\section{Introduction}\label{sec1}

One of the main concerns in the business world is to make a decision in determining the prices of products to be sold in order to maximize the companies' revenue. On the other hand, there is concern about the feasibility of purchasing these products by consumers, since there are several segments of potential buyers with different purchasing powers. Thus, a challenge for the seller is to set prices for the products in order to reach the maximum number of consumers capable of acquiring them, thus increasing its revenue, while these prices fit into the consumers' budget.

This decision-making process also has computational challenges, and was addressed by \cite{Guruswami2005}, who formulated the \textit{envy-free pricing problem} in order to computationally model and solve this situation. 

To state the envy-free pricing problem, assume that there is a set $I $ of $m$ consumers and a set $J $ of $n $ different items. Each item has $c _j $ copies, and then a {\it supply vector} is given by $ \mathbf{c} = (c_1,c_2,\ldots, c_n) \in \mathbb{N}^n $. Each consumer has a {\it valuation}  $v_i(S)$ for each bundle $S \subseteq J$ of items (for convenience, it is assumed that $v _i (\emptyset) = 0 $ for all bidder $i \in I$); the $m \times 2^n $ matrix  of valuations is denoted by $V \in \reais_{\geq 0}^{m \times 2^n}$. Given a {\it price vector} $\mathbf{p} = (p_1,p_2,\ldots,p_n) \in \reais_{\geq 0}^n$, the {\it utility} that consumer $i$ derives from bundle $S$ is $U_i(S) = v_i(S) - p_S$,
where $p_S=\sum_{j \in S} p_j$. If consumer $i$'s utility for the bundle $S$ is non-negative, it is said that $S$ is {\it feasible} for $i$. Consumer $i$'s {\it demand set} $D_i$ contains the bundles that will make him more satisfied. Formally, $D_i = \arg \max_{S \subseteq J} U_i(S)$. Since not buying any bundle is always an option with utility $U_i(\emptyset) = 0$, it follows that $U_i(S) \geq 0$ for all $S \in D_i$.

Using this terminology, we can define: an allocation  $(S_1,S_2,\ldots,S_m)$ of bundles to consumers is {\it feasible} if each item $j$ is in at most $c_j$ sets $S_i$;  iven a pricing $\mathbf{p} = (p_1,p_2,\ldots,p_n)$, an allocation
	$(S_1,S_2,\ldots,S_m)$ is {\it envy-free} if $S_i \in D_i$ for all $i$, i.e., each
	consumer receives a bundle from his demand set; a pricing $\mathbf{p}$ is {\it envy-free} if it admits a feasible, envy-free
	allocation.

In short, the envy-free pricing problem can be stated as follows: given the input $(m,n,V,\mathbf{c})$, compute an envy-free pricing $\mathbf{p}$ and a corresponding envy-free allocation $(S_1,S_2,\ldots,S_m)$ maximizing the seller's revenue  $\sum_{i=1}^{m} p_{S_i}$. \cite{Guruswami2005} proved that the envy-free pricing problem is APX-hard even if each item exists in unlimited supply, and each consumer has equal valuations (of either 1 or 2) for all
the items he has any interest in.

Similarly to \cite{Arbib2019}, we are interested in the particular case where $| I |=| J | =n$, $c_j=1$ for all $j$, and $v_i(S)>0$ only when $| S |=1$ for all $i$, i.e., each consumer can buy exactly one item (this is the \textit{unit-demand} case  and each item has only one unit available, what they called the {\it envy-free perfect matching problem}. Therefore, discarding the empty set and the subsets $S\subseteq J$ such that $| S| >1$, the dimension of the valuations matrix reduces from $m \times 2 ^ n $ to $n \times n$, and the input $v _i (S) = v_{ij} $ denotes the value assigned by the bidder $i $ to the item $j $. The choice of this particular case is due to the fact that these authors proved that, with these characteristics, the problem can be solved in polynomial time, and therefore the development of increasingly efficient algorithms to solve it is a research question to be addressed.

\subsection{Related work}

After the presentation of the envy-free pricing problem by Guruswami et al. \cite{Guruswami2005}, several authors have sought to develop polynomial-time algorithms for special optimal cases. \cite{gunluk2008pricing} considered the case of non-uniform demand case, unlimited supply and the valuations matrix satisfying the Monge property, presenting an algorithm that solves the problem optimally in $O(mn^2)$ time.

In the case of unit-demand with unlimited supply, and each consumer has a travel cost if he wants to purchase the same type of product in a different location from his own, so that the costs of all consumers form a metric space, \cite{chen2011optimal} presented an algorithm that solves the problem optimally in $O(n^4)$ time. In the  case of unit-demand and the number of consumers coincides with the number of items for sale,  \cite{Arbib2019} presented an algorithm that solves the problem optimally in $O(n^3)$ time.

More recently,  \cite{marcosms2021a14100279} considered the same case studied by \cite{chen2011optimal}, but presenting an algorithm that solves the problem optimally in $O(n^3)$ time.

\subsection{Our contribution}
\cite{Arbib2019} proved that the envy-free perfect matching problem can be solved in cubic time, using maximum weight perfect matchings to find optimal envy-free allocations, and shortest paths to find optimal envy-free prices. In this work we propose to modify the way of finding the envy-free prices, using the consumers' utilities in the construction of numerical sequences that converge to the maximum utilities, in order to obtain the optimal envy-free prices. Keeping this in mind, we propose a procedure for theoretical purposes and an efficient version of it for empirical performance purposes. 

We proved that at least one consumer reaches its maximum utility at each loop iteration of the method, thus exploring an optimal substructure of the problem, under this new approach. In addition, the vector/matrix structure used by us to calculate maximum utilities and, consequently, optimal envy-free prices, showed to be more efficient than the network used by \cite{Arbib2019}, when implemented.

 The computational experiments showed that our approach had better time performance in the search for optimal solutions, with a gain of approximately 48\% compared to the shortest paths approach.

%%%%%%%%%%%%%%%%%%%%%%%%%%%%%%%%%%%%%%%%%%%%%%%%%%%

\section{Mathematical foundation}
\label{sec:math}

The {\it First Social Welfare Theorem} \cite{Nisan:2007:AGT:1296179} ensures that an allocation that ma\-xi\-mizes the social welfare, i.e.,  the sum of valuations of the buyers for their allocations, also maximizes the revenue with envy-free prices. Thus, finding optimal allocations is equivalent to finding maximum weight perfect matchings in the graph associated with matrix $V$. Let $X=(x_{ij}) \in \{0,1\}^{n\times n}$ be the binary matrix associated with such a matching. We have the following:
\begin{lemma} \label{lemmamaxdiag}
	Let $\{i_1,i_2,\ldots, i_r\}$ be a subset of $I$ and $V'= X^TV =(v'_{ij})$. Then 
	$
	v'_{i_2i_1} + v'_{i_3i_2} + \ldots + v'_{i_r i_{r-1}} + v'_{i_1i_r} \leq v'_{i_1i_1} + v'_{i_2i_2} + \ldots +v'_{i_{r-1} i_{r-1}} +v'_{i_ri_r}.$
	
\end{lemma}

\begin{proof}
	Otherwise, there would be a permutation matrix $\Pi$ such that $\pi_{i_2i_1}=\pi_{i_3i_2}=\ldots=\pi_{i_ki_{k-1}}=\pi_{i_1i_k}=1,$ and then $tr(\Pi^TV) > tr(X^TV)$, which contradicts the optimality of X.
\end{proof}

Once an optimal allocation is found, consider the matrix $\mathcal{U}=(u_{ij})\in\reais^{n\times n}$ such that $u_{ij}=v'_{ij}-v'_{jj}$. This matrix inherits the fundamental property of matrix $V'$ shown in Lemma \ref{lemmamaxdiag}, which is stated in the following:
\begin{corollary}\label{corollarymaxdu}
	Let $\{i_1,i_2,\ldots, i_r\}$ be a subset of $I$. Then $$	u_{i_1 i_r} + \displaystyle \sum_{k=1}^{r-1} u_{i_{k+1} i_k} \leq 0 $$
	
\end{corollary}

\begin{proof}
	Matrix $\mathcal{U}$ is obtained by adding (negative) constants to each column of $V'$, which would introduce the constants relative to columns $i_1,i_2,\ldots, i_r$ on both sides of the inequality in Lemma \ref{lemmamaxdiag}. Since the diagonal elements of $\mathcal{U}$ are zero, the result follows.
\end{proof}

Due to space limitations, the proofs of Lemma \ref{lemmamaxdiag} and Corollary \ref{corollarymaxdu} were omitted.

The following lemma relates the matrix $\mathcal{U}$ to the buyers' utilities:
\begin{lemma}\label{lemmautilities}
	If $y_1, y_2, \ldots, y_n $ are minimal nonnegative numbers such that
	\begin{equation} \label{utilities}
	y_j \geq y_k + u_{jk}, \ \forall j,k\in I, j\neq k,
	\end{equation}
	then these numbers are the buyers' maximum utilities.
\end{lemma}
\begin{proof}
	Indeed, putting $p_j=v'_{jj}-y_j,\ \forall j \in I$, we have 
	\begin{equation*}
	v'_{jj}-p_j = y_j \geq y_k + u_{jk} = v'_{kk} - p_{k} + v'_{jk} - v'_{kk} = v'_{jk} - p_k, \forall  j \neq k,
	\end{equation*}
	and the result follows.
\end{proof}

For each $j \in I$, consider the sequence $(y_j^{(t)})_{t \in \mathbb{N}}$ defined as follows:
\begin{equation}\label{sequence}
y_j^{(0)}=0,\ y_j^{(t+1)}=\displaystyle\max_{k \in I} \{y_k^{(t)}+u_{jk}\}.
\end{equation}

\begin{theorem}\label{theoremconverge}
	The sequences defined in (\ref{sequence}) are convergent, and converge in at most $n-1$ steps.
\end{theorem}
\begin{proof}
	To visualize this fact, construct a tree rooted at node 0, and the other nodes are created by the following procedure: after the first step, if $y_j^{(1)}>0$, then node $j$ is a child of the root. In the next steps, if 
	\begin{equation*}
	y_k^{(t)}+u_{jk}=\displaystyle\max_{l\in I} \{y_l^{(t)}+u_{jl}\},
	\end{equation*} then node $j$ is a child of node $k$ that is on level $t$ (if there is another $k'$ that satisfies this condition, choose the smaller, without affecting the procedure).

	Note that a path from any node $j$ to root contains only distinct nodes, because if it were not so, assuming there is a path containing the repetition $\{j_1,j_2,\ldots,j_r,j_1\}$. We would have $y_{j_2}^{(s)}<y_{j_1}^{(s)}+u_{j_2 j_1}$, $y_{j_3}^{(s+1)}<y_{j_2}^{(s+1)}+u_{j_3 j_2}=y_{j_1}^{(s)}+u_{j_2 j_1}+u_{j_3 j_2}$, $\ldots$, $y_{j_r}^{(s+r-3)}<y_{j_{r-1}}^{(s+r-3)}+u_{j_r j_{r-1}}=y_{j_1}^{(s)}+u_{j_2 j_1}+u_{j_3 j_2}+\ldots+u_{j_r j_{r-1}}$,  $y_{j_1}^{(s+r-2)}<y_{j_{r}}^{(s+r-2)}+u_{j_1 j_r}=y_{j_1}^{(s)}+u_{j_2 j_1}+u_{j_3 j_2}+\ldots+u_{j_r j_{r-1}}+u_{j_1 j_r}$,

\noindent for some $s \in \mathbb{N}$, and then $u_{j_2 j_1}+u_{j_3 j_2}+\ldots+u_{j_r j_{r-1}}+u_{j_1 j_r}>y_{j_1}^{(s+r-2)}-y_{j_1}^{(s)}\geq0,$ which, by Corollary \ref{corollarymaxdu}, is impossible. Hence, the height of the tree is at most $n-1$, and therefore
	\begin{equation}\label{equationmaxsteps}
	y_j^{(t+1)}= y_j^{(t)},\quad  \forall j \in I,\ \forall t \geq n-1,
	\end{equation}
	from where the result follows.
\end{proof}
\begin{corollary}\label{corollaryj0}
	There is at least one $j_0 \in I$ such that 
	\begin{equation*}
	y_{j_0}^{(t)} = 0,\quad \forall t \in \mathbb{N}.
	\end{equation*}
\end{corollary}
\begin{proof}
	Indeed, since the paths in the tree constructed in proof of Theorem \ref{theoremconverge} contains no repetitions, there is at least one $j_0 \in I$ such that node $j_0$ is not in it, and the result follows.
\end{proof}

\begin{corollary} \label{corolario:maxlocal}
    For each $t=0, \ldots, n-1$, there is a $j_t$ such that $y_{j_t}^{(t)}$ is maximum.
\end{corollary}

\begin{proof}
    Otherwise, there would be repetition of nodes in some path of the tree constructed in the proof of Theorem \ref{theoremconverge} and, by Corollary \ref{corollaryj0}, there is $j_0 \in I$ such that $\max_{t=0, \ldots, n-1} \{ y_{j_0}^{(t)} \} = 0$, and the result follows.
\end{proof}

\begin{corollary}\label{corollaryminimal}
	The sequences defined in (\ref{sequence}) converge to minimal nonnegative numbers $y_1,y_2,\ldots,y_n$ satisfying (\ref{utilities}).
\end{corollary}
\begin{proof}
	Suppose there are nonnegative numbers $z_1,z_2, \ldots, z_n$ satisfying
	\begin{equation*}
	z_j \geq z_k + u_{jk}, \ \forall j,k\in I, j\neq k,
	\end{equation*}
	with $z_j \leq y_j,\ \forall j \in I$.
	Thus, for each $j \in I$, using (\ref{equationmaxsteps}), we have $y_j=y_j^{(n)} = y_{j_1}^{(n-1)}+u_{j j_1}$ and $ z_j \geq z_{j_1}+u_{j j_1}$, $y_{j_1}^{(n-1)} = y_{j_2}^{(n-2)}+u_{j_1j_2}$ and $ z_{j_1} \geq z_{j_2}+u_{j_1 j_2}, \ldots$, $y_{j_{n-1}}^{(1)} = y_{j_n}^{(0)}+u_{j_{n-1}j_n} $ and $ z_{j_{n-1}} \geq z_{j_n}+u_{j_{n-1} j_n}.$

	Subtracting the equations in the left from inequalities in the right, we obtain $z_j-z_{j_1} \geq y_j - y_{j_1}^{(n-1)},$ $z_{j_1}-z_{j_2} \geq y_{j_1}^{(n-1)} - y_{j_2}^{(n-2)},\ldots$ $z_{j_{n-1}}-z_{j_n} \geq y_{j_{n-1}}^{(1)} - y_{j_n}^{(0)}.$

	Summing all these inequalities, it follows that
	\begin{equation*}
	z_j - z_{j_n} \geq y_j - y_{j_n}^{(0)}.
	\end{equation*}
	Since $z_{j_n} \geq 0$ and $y_{j_n}^{(0)}=0$, we then have
	\begin{equation*}
	z_j \geq z_j - z_{j_n} \geq y_j - y_{j_n}^{(0)} \geq y_j.
	\end{equation*}
	Therefore, $z_j = y_j \ \forall j \in I$, and the result follows.
\end{proof}

\begin{theorem}\label{theoremmaxprices}
	Let $\mathbf{p}=(p_1,p_2,\ldots,p_n)$ be such that
	\begin{equation}\label{equationprices}
	p_j = v'_{jj} - y_j, \quad \forall j \in I.
	\end{equation}
	Then $\mathbf{p}$ is an envy-free price vector that maximizes the seller's revenue.
\end{theorem}
\begin{proof}
	By Corollary \ref{corollaryminimal}, the nonnegative numbers $y_1,y_2,\ldots,y_n$ are the minimal satisfying (\ref{utilities}). Thus, putting the prices as in (\ref{equationprices}), Lemma \ref{lemmautilities} ensures that they maximize the buyers' utilities. Moreover, 
	\begin{equation*}
	y_i \geq y_j + u_{ij}=y_j+v'_{ij}-v'_{jj} \Rightarrow y_j - y_i +v'_{ij} \leq v'_{jj}, \quad \forall i,j \in I,\ i\neq j.
	\end{equation*}
	Fixing $i\in I$ such that $y_i=0$, whose existence is guaranteed by Collorary \ref{corollaryj0}, we have
	\begin{equation*}
	y_j +v'_{ij} \leq v'_{jj} \Rightarrow y_j \leq v'_{jj},\quad \forall j \in I.
	\end{equation*}
	Therefore, the items allocated are in buyers' demand sets, and the result follows.
\end{proof}

%%%%%%%%%%%%%%%%%%%%%%%%%%%%%%%%%%%%%%%%%%%%%%%%%%%%%%

\section {Method for finding the envy-free prices}
\label{sec:algorithm}

Now, we will present a method that we will call \textsc{PricesEFPM}  that computes envy-free prices that maximizes the seller's revenue. The method is described in the procedure bellow:

\begin{enumerate}
    
    \item Construct the matrix $\mathcal{U}$ such that $u_{ij} = v'_{ij}-v'_{jj}$. 
    
    \item The initial sequences of utilities vector is defined as $\mathbf{y}^{(0)} = (0,0,\ldots,0)  $.
    
    \item The next sequences of utilities vector is defined as $y_j^{(1)} = \displaystyle\max_{k \in I} \{u_{jk}\} $.
    
    \item While $y_j^{(t)} > y_j^{(t-1)}$ for some $j $ do $y_j^{(t+1)} = \displaystyle\max_{k = 1,\ldots,n} \{y_k^{(t)}+u_{jk}\}$.

    \item Set the prices $p_j = v'_{jj}-y_j^{(t)}$.
\end{enumerate}

This procedure computes the envy-free prices that maximizes the seller's revenue. Intuitively, Step 1 builds the matrix $\mathcal{U}$, while the sequence of utilities starts with zeros in Step 2, followed by maximum utilities in Step 3. From there, Step 4 continue to build the sequence of utilities until a maximum is reached according to the Lemma~\ref{lemmautilities}, and the loop ends due to Theorem~\ref{theoremconverge}. Finally, the envy-free optimal prices that maximize revenue are calculated in Step 5, according to Theorem~\ref{theoremmaxprices}.

Next are the theoretical proofs of  correctness of the algorithm, as well as its time complexity.

\begin{theorem}(invariant)
	For all $t\in \mathbb{N}$, holds
	\begin{equation} \label{equationinvariant}
	y_j^{(t+1)} \geq y_k^{(t)}+u_{jk},\quad \forall j,k \in I,\ j\neq k.
	\end{equation}
\end{theorem}
\begin{proof}
	Indeed, Step 3 defines
	$y_j^{(1)} =\displaystyle\max_{l \in I} \{u_{jl}\}
	=  \displaystyle\max_{l \in I} \{y_l^{(0)}+u_{jl}\}
	 \geq  y_k^{(0)}+u_{jk}, \forall  j\neq k.
	$

	Thus, the invariant given by inequalities (\ref{equationinvariant}) holds prior to the first while loop iteration, des\-cri\-bed in Step 4. Moreover, if stopping criterion is not satisfied, then  	
	$
	y_j^{(t+1)} =  \displaystyle\max_{l \in I} \{y_l^{(t)}+u_{jl}\} \geq y_k^{(t)}+u_{jk}, \forall j\neq k,
	$	
	and so, the invariant also holds at the end of an arbitrary while loop i\-te\-ra\-tion, just like in the beginning, and the result follows.
\end{proof}

\begin{theorem}(termination) The while loop stops in at most $n-1$ iterations.
\end{theorem}
\begin{proof}
	By Theorem \ref{theoremconverge}, holds
	\begin{equation*}
	y_j^{((n-1)+1)} = y_j^{(n)} = y_j^{(n-1)},\quad \forall j \in I.
	\end{equation*}	
\end{proof}

Since $t=1$ at beginning of while loop (Step 4), the result follows.
\begin{theorem}(correctness for utilities)
	At the end of while loop, holds
	\begin{equation}\label{equationstoploop}
	y_j^{(t)}\geq y_k^{(t)} + u_{jk},\quad \forall j,k \in I, j\neq k,
	\end{equation}
	for some $t \in \{1,\ldots, n\}$.
\end{theorem}
\begin{proof}
	Indeed, when stopping criterion in Step 4 is satisfied, then (\ref{equationstoploop}) holds because 	
	 $y_j^{(t)} = y_j^{(t-1)} \Rightarrow  y_j^{(t)} = \displaystyle\max_{l \in I} \{y_l^{(t-1)}+u_{jl}\} 
	= \displaystyle\max_{l \in I} \{y_l^{(t)}+u_{jl}\}
	 \geq y_k^{(t)} + u_{jk}, $	 
	 for all $ j\neq k.$
\end{proof}

\begin{theorem}(correctness for prices) 
	The procedure returns an envy-free price vector $\mathbf{p}$ that maximizes the seller's revenue.
\end{theorem}
\begin{proof}
	Step 5 produces a vector $\mathbf{p}$ that satisfies (\ref{equationprices}). Therefore, by Theorem \ref{theoremmaxprices}, the result follows.
\end{proof}

\begin{theorem}(time complexity)
	The procedure runs in $O(n^3)$ time.
\end{theorem}
\begin{proof}	
	The highest computational cost corresponds to the while loop. This loop performs $n$ checks of $O(n)$ maximum operations in a vector, so in $O(n^2)$ time, and  performs these $O(n^2)$ operations in at most $n-1$ times. Therefore, the result follows.
\end{proof}

%%%%%%%%%%%%%%%%%%%%%%%%%%%%%%%%%%%%%%%%%%%%%%

\section{Computational experiments} \label{sec:experiments}

Our algorithms were evaluated on several randomly generated datasets using a 64-bit Intel Core i3-6006U computer running a clock of 2.0 GHz with 4 GB of RAM.
The codes were implemented in C++ and compiled in g++ 8.1.0. We generated 15 random test cases for $n=1,000,\ n=2,000,\ n=5,000,\ n=10,000$ and $n=15,000$. In all random test cases, each $v_{ij}$ is a random integer chosen from a uniform distribution on the integers between $0$ and $1,000,000$. In implementing the Bellman-Ford algorithm, we removed the detection of negative cycles part, because if an optimal allocation is given as input, these cycles do not exist in the network that models the pricing subproblem \cite{Arbib2019}. The open source codes are available in the GitHub repository, at \url{https://github.com/marcosmsgithub/efpm.git}.

Table \ref{tablealgo} shows the average running times and the $95\%$ confidence intervals (CI), in seconds with two-digit precision, of the Bellman-Ford algorithm as well as of Algorithm \textsc{PricesEFPM}. As we expected, Algorithm \textsc{PricesEFPM} ran faster than the Bellman-Ford algorithm in all simulations. For n=1,000, Algorithm \textsc{PricesEFPM} obtained an average time reduction of 52,2\% in determining optimal envy-free prices compared to the Bellman-Ford algorithm. For n=2,000, the average time reduction was approximately 46,3\%. For n=5,000, n=10,000 and n=15,000, the average time reductions were approximately 51,6\%, 43,8\% and 44,8\%, respectively. Thus, in general, Algorithm \textsc{PricesEFPM} achieved an average time reduction in determining optimal envy-free prices of 47.7\% compared to the Bellman-Ford algorithm.

Furthermore, it is noted that the amplitudes of the 95\% confidence intervals of Algorithm \textsc{PricesEFPM} are smaller than those of the Bellman-Ford algorithm, showing Algorithm \textsc{PricesEFPM} to be more stable in the experiments performed, providing better statistical precision in the analyses.

	\begin{table}[!h] \centering
	\caption{Comparison between performance running times of the algorithms, with average times in seconds and confidence intervals of 95\%}  
		\label{tablealgo}
		\begin{adjustbox}{width=0.8\textwidth}

		\begin{tabular}{ccccc}
			\hline			
			\multicolumn{1}{c}{n} & \multicolumn{2}{c}{Bellman-Ford} & \multicolumn{2}{c}{ \textsc{PricesEFPM}} \\
			
			\cmidrule{2-5}

			& \multicolumn{1}{c}{Avg. time} & \multicolumn{1}{c}{$95\%$ CI} & \multicolumn{1}{c}{Avg. time} & \multicolumn{1}{c}{$95\%$ CI} \\
			
			\hline 
			 
			1,000 & 0.69 &  (0.56,0.74) &  0.33 & (0.28,0.38) \\

			2,000 & 3.89 &   (3.49,4.30) &   2.09 &  (1.80,2.39) \\

			5,000 & 53.82 &   (47.16,60.48) &  26.07 & (22.75,29.39) \\

			10,000 & 312.16 &   (286.39,337.94) &  175.39 & (161.80,188.81) \\

			15,000 &  864.60 &  (777.08,952.11) &  477.21 &  (431.19,523.22) \\			
				
			\hline

		\end{tabular}

\end{adjustbox}
           
\end{table}

%%%%%%%%%%%%%%%%%%%%%%%%%%%%%%%%%%%%%%%%

\section{Conclusions}
\label{sec:conclusions}
We discussed the envy-free perfect matching problem by focusing on the pri\-cing subproblem. For the allocation subproblem, theoretical results guarantee that this can be solved as a maximum weight perfect matching problem, which is already well explored, and several $O(n^3)$ algorithms have already been de\-ve\-lo\-ped. We have been able to design an algorithm that uses a vector/matrix structure that allows faster computations in the search of envy-free prices than the network used by Bellman-Ford algorithm to find the envy-free prices as shortest paths. 

According to experiments performed, the algorithm proposed in this work obtained an average time reduction in determining optimal envy-free prices of at least 44\%, and at most 52\%, approximately, compared to the algorithm found in the literature. Even for the largest tested instances, when n=15,000, the reduction was approximately 45\%, and overall the reduction was approximately 48\%.

However, it should be noted that the strategy used by the proposed method applies only if the maximum weight perfect matching given as input is exact. If the matching is an approximation of the optimal, then the algorithm will go into an infinite loop, as the procedure will violate the structural condition of matrix ${\cal U}$ presented in Corollary \ref{corollarymaxdu}. This limitation encourages the design of efficient approximation algorithms for finding envy-free prices in envy-free perfect matchings when dealing with very large instances.




\bibliographystyle{unsrtnat}
\bibliography{references}  %%% Uncomment this line and comment out the ``thebibliography'' section below to use the external .bib file (using bibtex) .





% \end{thebibliography}

\end{document}